\documentclass{article}

\usepackage[hmargin=3.5cm,vmargin=3.5cm,a4paper]{geometry}

\usepackage{graphicx}
\usepackage{mathtools}
\usepackage{amssymb}
\usepackage{wrapfig}
\usepackage{tikz}
\usetikzlibrary{trees}
\usepackage{xcolor}
\usepackage{amsthm}
\newtheorem{theorem}{Theorem}

\newtheorem{lemma}[theorem]{Lemma}

\newtheorem{corollary}[theorem]{Corollary}
\newtheorem{proposition}[theorem]{Proposition}
\theoremstyle{definition}
\newtheorem{definition}[theorem]{Definition}
\theoremstyle{remark}
\newtheorem{example}[theorem]{example}

\newcommand{\field}[1]{\mathbb{#1}}
\newcommand{\R}{\field{R}} 
\newcommand{\N}{\mathcal{N}} 

\newcommand{\Ga}{G}
\newcommand{\Gr}{\Gamma}

\newcommand{\A}{\mathcal{A}}
\newcommand{\Pro}{\mathcal{O}}
\newcommand{\LPoA}{\textsf{LPoA}}
\newcommand{\SPoA}{\textsf{SPoA}}
\newcommand{\PoA}{\textsf{PoA}}
\newcommand{\Rosenthal}{congestion }
\newcommand{\PoS}{\textsf{PoS}}
\newcommand{\NE}{\text{$\mathbb{NE}$}}
\newcommand{\SE}{\text{$\mathbb{SPO}$}}
\newcommand{\SPO}{\text{$\mathbb{SPO}$}}

\definecolor{mred}{rgb}{1,0,0}
\definecolor{mgreen}{rgb}{0,0.56078,0.17647}
\definecolor{mblue}{rgb}{0.141176,0.188235,0.47843}
\definecolor{mmagenta}{rgb}{0.745098,0.386235,0.643137}
\definecolor{morange}{rgb}{0.949019,0.572549,0.0980392}
\definecolor{mpurple}{rgb}{0.4196078,0.30196,0.6117647}
\usepackage{authblk}

\newcommand{\rowcolor}[1]{\textcolor{red}{#1}}
\newcommand{\columncolor}[1]{\textcolor{blue}{#1}}

\renewcommand{\l}{\ell}

\DeclareMathOperator*{\argmin}{argmin}
\DeclareMathOperator{\Id}{Id}

\title{\bfseries The Curse of Ties in Congestion Games with \\ Limited Lookahead}

\author[1]{Carla Groenland}
\affil[1]{University of Oxford, \texttt{groenland@maths.ox.ac.uk}}
\author[2]{Guido Sch\"afer}
\affil[2]{CWI and Vrije Universiteit Amsterdam, \texttt{g.schaefer@cwi.nl}}

\begin{document}
\maketitle


\begin{abstract}  
We introduce a novel framework to model limited lookahead in congestion games. Intuitively, the players enter the game sequentially and choose an optimal action under the assumption that the $k-1$ subsequent players play subgame-perfectly. Our model naturally interpolates between outcomes of greedy best-response ($k=1$) and subgame-perfect outcomes ($k=n$, the number of players). We study the impact of limited lookahead (parameterized by $k$) on the stability and inefficiency of the resulting outcomes. As our results reveal, increased lookahead does not necessarily lead to better outcomes; in fact, its effect crucially depends on the existence of ties and the type of game under consideration.

More specifically, already for very simple network congestion games we show that subgame-perfect outcomes (full lookahead) can be unstable, whereas greedy best-response outcomes (no lookahead) are known to be stable. We show that this instability is due to player indifferences (ties). If the game is generic (no ties exist) then all outcomes are stable, independent of the lookahead $k$. In particular, this implies that the price of anarchy of $k$-lookahead outcomes (for arbitrary $k$) equals the standard price of anarchy. For special cases of cost-sharing games and consensus games we show that no lookahead leads to stable outcomes. Again this can be resolved by removing ties, though for cost-sharing games only full lookahead provides stable outcomes. We also identify a class of generic cost-sharing games for which the inefficiency decreases as the lookahead $k$ increases.
\end{abstract}

%



\newpage



\section{Introduction}
\setlength\intextsep{10pt}
\begin{wrapfigure}{r}{2.5cm}
    \begin{tikzpicture}
        \begin{scope}
        \small
            \node at (0,0) [circle,fill=black, scale = 0.4] (s) {};
            \node at (0,-0.25) [] {$o$};
            \node at (0,3) [circle,fill=black, scale = 0.4] (t) {};
            \node at (0,3.25) [] {$d$};
            \node at (0,1.5) [circle,fill=black, scale = 0.4] (m) {};
            
            \draw (s) edge[->,thick, bend right=60] (t);
            \node at (-0.2,0.75) [] {$b$};
            \node at (-0.55,2.25) [] {$\ell$};
            \node at (0.2,2.25) [] {$s$};
            \node at (1,1.5) [] {$m$};
            \draw (s) edge[->,thick] (m);
            \draw (m) edge[->,thick, bend left=50] (t);
            \draw (m) edge[->,thick] (t);
        \end{scope}
    \end{tikzpicture}
\end{wrapfigure}

Consider the following 
situation, where two players want to travel from origin $o$ to destination $d$ in the (extension-parallel) graph on the right. They can take the metro $m$, which takes 6 minutes, or they can take the bike $b$ and then walk either the long (but scenic) route $\ell$, which takes $2$ minutes, or the short route $s$, which takes $1$ minute. There is only one bike: if only one of them takes the bike it takes $3$ minutes; otherwise, someone has to sit on the backseat and it takes them $5$ minutes. 
Both players want to minimize their own travel time. 

Suppose they announce their decisions sequentially. There are two possible orders: either the \rowcolor{red} player 1 moves first or the \columncolor{blue} player 2 moves first. We consider the \emph{sequential-move version} of the game where player $1$ moves first. 
There are three possible \textit{subgames} that player 2 may end up in, for which the corresponding game trees are as follows: 
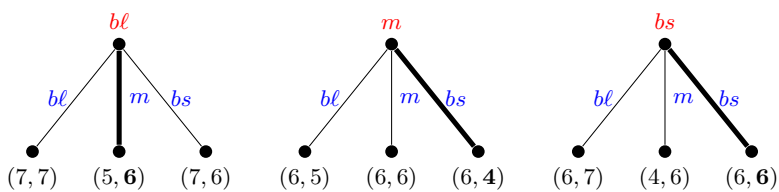
\begin{center}
\hspace*{-.3cm}
\scalebox{.95}{
\small
\tikzstyle{level 1}=[level distance=1.5cm, sibling distance=1.2cm]
\tikzstyle{bag} = [circle, minimum width=5pt,fill, inner sep=0pt]
\tikzstyle{end} = [circle, minimum width=5pt,fill, inner sep=0pt]
\begin{tikzpicture}[]
\node[bag, label=above:{\rowcolor{$b \ell$}}] {}
    child {
        node[end, black,label=below:
            {$(7,7)$}] {}
        edge from parent[black]
        node[left] {\columncolor{$b \ell$}}
    }
    child {
        node[end, black, label=below:
            {$(5,\textbf{6})$}] {}
        edge from parent[line width=2pt]
        node[right] {\columncolor{$m$}}
    }
    child {
            node[end, label=below:
                {$(7,6)$}] {}
            edge from parent[black]
            node[right] {\columncolor{$bs$}}
    };
\end{tikzpicture}
\quad
\begin{tikzpicture}[]
\node[bag, label=above:{\rowcolor{$m$}}] {}
    child {
        node[end, black,label=below:
            {$(6,5)$}] {}
        edge from parent[black]
        node[left] {\columncolor{$b \ell$}}
    }
    child {
        node[end, black, label=below:
            {$(6,6)$}] {}
        edge from parent[]
        node[right] {\columncolor{$m$}}
    }
    child {
            node[end, label=below:
                {$(6,\textbf{4})$}] {}
            edge from parent[line width=2pt]
            node[right] {\columncolor{$bs$}}
    };
\end{tikzpicture}
\quad
\begin{tikzpicture}[]
\node[bag, label=above:{\rowcolor{$bs$}}] {}
    child {
        node[end, black,label=below:
            {$(6,7)$}] {}
        edge from parent[black]
        node[left] {\columncolor{$b\ell$}}
    }
    child {
        node[end, black, label=below:
            {$(4,6)$}] {}
        edge from parent[black]
        node[right] {\columncolor{$m$}}
    }
    child {
            node[end, label=below:
                {$(6,\textbf{6})$}] {}
            edge from parent[line width=2pt]
            node[right] {\columncolor{$bs$}}
    };
\end{tikzpicture}
}
\end{center}

A \textit{strategy} for player 2 is a function
$
S_2:\{\rowcolor{b\ell},\rowcolor{bs},\rowcolor{m}\}\to\{\columncolor{b\ell},\columncolor{bs},\columncolor{m}\}
$ 
that tells us which action player 2 plays given the action of player 1. Player 2 will always choose an action that minimizes his travel time and may break ties arbitrarily when being indifferent. The boldface arcs give a possible subgame-perfect strategy for player 2. 
If player $2$ fixes this strategy, then player 1 is strictly better off taking the bike and walking the long route (for a travel time of $5$ compared to a travel time of $6$ for the other cases). That is, the only subgame-perfect response for player 1 is $\rowcolor{b\ell}$. This shows $(\rowcolor{b\ell},\columncolor{m})$ is a \emph{subgame-perfect outcome}. However, this outcome is rather peculiar: Why would player 1 walk the long route if his goal is to arrive as quickly as possible? In fact, the outcome $(\rowcolor{b\ell},\columncolor{m})$ is not \emph{stable}, i.e., it does not correspond to a Nash equilibrium.

Subgame-perfect outcomes are introduced as a natural model for farsightedness \cite{leme,selten}, or ``full anticipation'', and have been studied for various types of congestion games \cite{anomalies,bartjasper,jong,leme}. Another well-studied notion in this context are outcomes of \emph{greedy best-response} \cite{fot06,optimalNEmakespan,economicspaper,optimalNE}, i.e., players enter the game one after another and give a best response to the actions played already, thus playing with ``no anticipation''. Fotakis et al. \cite{fot06} proved that greedy best-response leads to stable outcomes on all \emph{series-parallel graphs} (which contain extension-parallel graphs like the one above as a special case).

In fact, in the above example both $(\rowcolor{bs},\columncolor{m})$ and $(\rowcolor{bs},\columncolor{bs})$ are greedy best-response outcomes and they are stable. 
The example thus illustrates that full lookahead may have a negative effect on the stability of the outcomes. 
After a moment's thought, we realize that in the subgame-perfect outcome the indifference of player 2 is exploited (by breaking ties accordingly) to force player 1 to play a suboptimal action. 
Immediate questions that arise are: 
Does full lookahead guarantee stable outcomes if we adjust the travel times such that the players are no longer indifferent (i.e., if we make the game \emph{generic})? 
What is the lookahead that is required to guarantee stable outcomes? What about the inefficiency of these outcomes?
In this paper, we address such questions. 

A well-studied playing technique for chess introduced by Shannon \cite{shan} is to expand the game tree up to a fixed level, use an evaluation function to assign values to the leaves and then perform backward induction to decide which move to make. 
Based on this idea, we introduce \emph{$k$-lookahead outcomes} as outcomes that arise when every player uses such a strategy with $k$ levels of backward induction. 
The motivation for our studies is based on the observation that such limited lookahead strategies are played in many scenarios. 
In fact, there is also experimental evidence (see, e.g., \cite{lamotivation}) that humans perform limited backward induction rather than behaving subgame-perfectly or myopically (i.e., $1<k<n$). 

In general, our notion of $k$-lookahead outcomes can be applied to any game that admits a natural evaluation function for partial outcomes (details will be provided in the full version of this paper). In this paper, we demonstrate the applicability of our novel notion by focussing on congestion games. These games admit a natural evaluation function by assigning a partial outcome its current cost, i.e., the cost it would have if the game would end at that point.

\paragraph{Our model.} 

We introduce \emph{$k$-lookahead outcomes} as a novel solution concept where players enter the game sequentially (according to some arbitrary order) and anticipate the effect of $k-1$ subsequent decisions. 
In a \emph{$k$-lookahead outcome} $A = (A_1, \dots, A_n)$, 
the $i$th player with $i \ge 1$ computes a subgame-perfect outcome in the subgame induced by $(A_1, \dots, A_{i-1})$ with $\min\{k,n-i\}$ players (according to the order) and chooses his corresponding action. Our model interpolates between outcomes of greedy best-response ($k=1$) and subgame-perfect outcomes ($k=n$, the number of players). 
Our main goal is to understand the effect that different degrees of anticipation have on the stability and inefficiency of the resulting outcomes in congestion games. 

We combine limited backward induction with the approach of Paes Leme et al.~\cite{leme} who proposed to study the inefficiency of subgame-perfect outcomes. 
Subgame-perfect outcomes have several drawbacks as model for anticipating players, which are overcome (at least to some extent) by considering $k$-lookahead outcomes instead:
\begin{enumerate}
    \item \textit{Computational complexity}. 
    Computing a subgame-perfect outcome in a congestion game is PSPACE-complete \cite{leme} and it is NP-hard already for 2-player symmetric network congestion games with linear delay functions \cite{bartjasper}. This adds to the general discussion that if a subgame-perfect outcome cannot be computed efficiently, then its credibility as a predicting means of actual outcomes is questionable.
    On the other hand, computing a $k$-lookahead outcome for constant $k$ can be done efficiently by backward induction.
      \item \textit{Limited information}. Due to lack of information players might be forced to perform limited backward induction. Note that in order to expand the full game tree, a player needs to know his successors, the actions that these successors can choose and their respective preferences. In practice, however, this information is often available only for the first few successors. 
  \item \textit{Clairvoyant tie-breaking}. 
Players may be unable to play subgame-perfectly, unless some clairvoyant tie-breaking rule is implemented. To see this, consider the example introduced above. Note that in the subgame induced by $\rowcolor{b\ell}$ (as well as $\rowcolor{bs}$) player $2$ is indifferent between playing $\columncolor{m}$ and $\columncolor{bs}$. Thus, player 1 has no way to play subgame-perfectly with certainty: in order to do so he will need to correctly guess how player 2 is going to break ties. Such clairvoyant tie-breaking is not required for reaching a $k$-lookahead outcome.
\end{enumerate}

\paragraph{Our results.}
We study the efficiency and stability of $k$-lookahead outcomes. We call an outcome stable if it is a Nash equilibrium.
In order to assess the inefficiency of $k$-lookahead outcomes, we introduce the \emph{$k$-Lookahead Price of Anarchy ($k$-$\LPoA$)} which generalizes both the standard Price of Anarchy \cite{kouts} and the Sequential Price of Anarchy \cite{leme} (see below for formal definitions). 

Quantifying the $k$-Lookahead Price of Anarchy is a challenging task in general. In fact, even for the Sequential Price of Anarchy (i.e., for $k = n$) no general techniques are known in the literature. In this paper, we mainly focus on characterizing when $k$-lookahead outcomes correspond to  stable outcomes. 
As a result, our findings enable us to characterize when the $k$-Lookahead Price of Anarchy coincides with the Price of Anarchy. 
We show that this correspondence holds for congestion games that are structurally simple (i.e., symmetric congestion games on extension-parallel graphs), called \emph{simple} below.
Further, a common trend in our findings is that the stability of $k$-lookahead outcomes crucially depends on whether players do not or do have to resolve ties (generic vs.~non-generic games). 
Our main findings in this paper are as follows: 
\begin{enumerate}
    \item We show that for generic simple congestion games the set of $k$-lookahead outcomes coincides with the set of Nash equilibria for all levels of lookahead $k$. 
    As a consequence, we obtain that the $k$-\LPoA\ coincides with the Price of Anarchy (independently of $k$), showing that increased anticipation does not reduce the (worst-case) inefficiency. 
    On the other hand, we show that only full anticipation guarantees the first player the smallest cost, so that anticipation might be beneficial after all. 
    We also show that the above equivalence does not extend beyond the class of simple congestion games. (These results are presented in Section 3.)
    
    \item For non-generic simple congestion games, subgame-perfect outcomes my be unstable (as the introductory example shows) but we prove that they have optimal egalitarian social cost. For the more general class of series-parallel graphs, we prove that the congestion vectors of 1-lookahead outcomes coincide with those of global optima of Rosenthal's potential function. In particular, this implies that the $1$-$\LPoA$ is bounded by the Price of Stability. 
    (See Section 3.)
    
    \item We also study cost-sharing games and consensus games (see below for definitions). For consensus games, subgame-perfect outcomes may be unstable. In contrast, if players break their ties consistently all $k$-lookahead outcomes are optimal. 
    Similarly, for non-generic cost-sharing games we show that even in the symmetric singleton case subgame-perfect outcomes may be unstable. For both symmetric and singleton games this can be resolved by removing the ties.  
    We also observe a threshold effect with respect to the anticipation level. For generic symmetric cost-sharing games, we show $k$-lookahead outcomes are stable but guaranteed to be optimal only for $k=n$. For affine delay functions the $k$-$\LPoA$ is non-increasing (i.e., the efficiency improves with the anticipation).
    For generic singleton cost-sharing games, $k$-lookahead outcomes are only guaranteed to be stable for $k=n$. 
    (These results can be found in Section 4.)
\end{enumerate}

\paragraph{Related work.}
The idea of limited backward induction dates back to the 1950s \cite{shan} and several researchers in artificial intelligence (see, e.g., \cite{lookaheadAI}) investigated it in a game-theoretic setting.
Mirrokni et al. \cite{lookaheadVetta} introduce \textit{$k$-lookahead equilibria} that incorporate various levels of anticipation as well. Their motivation for introducing these equilibria is very similar to ours, namely to provide an accurate model for actual game play. 
However, their $1$-lookahead equilibria correspond to Nash equilibria rather than greedy best-response outcomes and none of the equilibria correspond to subgame-perfect outcomes. 
Moreover, lookahead equilibria are not guaranteed to exist. For example, Bilo et al. \cite{lookaheadBilo} show that symmetric singleton congestion games do not always admit 2-lookahead equilibria (for the ``average-case model''). 

Subgame-perfect outcomes are special cases of $n$-lookahead outcomes.
Paes Leme et al.~\cite{leme} generalize the Price of Anarchy notion to subgame-perfect outcomes and show that the \emph{Sequential Price of Anarchy} can be much lower than the Price of Anarchy if the game is generic. On the other hand, this does not necessarily hold if the game is non-generic (see, e.g.,  \cite{anomalies,bartjasper,crowdgames}). 




\section{Lookahead outcomes}
\label{sec:klo}

We formally define congestion games and introduce some standard notation. We then introduce our notion of $k$-lookahead outcomes and the inefficiency measures studied in this paper. Finally, we comment on the impact of ties and different player orders in these games.

\paragraph{Congestion games.}
A \textit{\Rosenthal game} is a tuple $\Ga = (\N, R,(\A_i)_{i\in \N},(d_r)_{r\in R})$ where $\N=[n]$ is a finite set of players, $R$ a finite set of \textit{resources}, $\A_i\subseteq 2^R$ the \emph{action set} of player $i$, and $d_r:\mathbb{N} \to \R_{\geq 0}$ a \textit{delay function} for every resource $r\in R$.\footnote{We use $[n]$ to denote the set $\{1,\dots,n\}$, where $n$ is a natural number.} 
Unless stated otherwise, we assume that $d_r$ is non-decreasing.
We define $\Pro = \prod_{i \in \N} \A_i$ as the set of \emph{action profiles} or \emph{outcomes} of the game. 
Given an outcome $A\in \Pro$, the \textit{congestion vector} $x(A)=(x(A)_r)_{r\in R}$ specifies the number of players picking each resource, i.e., $x(A)_r = |\{i\in \N: r\in A_i\}|$.
The cost function $c_i$ of player $i\in \N$ is given by
$c_i(A)=\sum_{r\in A_i} d_r(x(A)_r)$. We call a congestion game \textit{symmetric} if $\A_i=\A_j = \A$ for all $i, j\in \N$.

We say that $A_i$ is a \textit{best response} to $A_{-i}$ if 
$c_i(A_i, A_{-i}) \le c_i(B_i, A_{-i})$ 
for all $B_i\in \A_i$.\footnote{We use the standard notation $A_{-i}=(A_1,\dots,A_{i-1},A_{i+1},\dots,A_n)$ and $(B_i,A_{-i})=(A_1,\dots,A_{i-1},B_i,A_{i+1},\dots,A_n)$.}
An outcome $A$ of a \Rosenthal game $\Ga$ is a \emph{(pure) Nash equilibrium (NE)} if for all $i\in \N$, $A_i$ is a best response to $A_{-i}$. 
We use $\NE(\Ga)$ to denote the set of all Nash equilibria of $\Ga$.


An \textit{order} on the players is a bijection $\sigma:\N\to[n]$. We denote the sequential-move version of a game $\Ga$ with respect to order $\sigma$ by $\Ga^\sigma$. The outcome on the equilibrium path of a subgame-perfect equilibrium in $\Ga^\sigma$ is an action profile of $\Ga$ and we refer to it as the \textit{subgame-perfect outcome (SPO)}. 
We use $\SPO(\Ga)$ to refer to the set of all subgame-perfect outcomes (with respect to any order of the players) of a game $\Ga$.

\paragraph{$k$-lookahead outcomes.}

Let $\Ga=(\N,R,(\A_i)_{i\in \N},(d_r)_{r\in R})$ be a congestion game and $\sigma$ an order on the players. For $k\in[n]$, define $\Ga^k(\sigma) =(\N',R,(\A_i)_{i\in \N'},(d_r)_{r\in R})$ as the congestion game with $\N'=\sigma^ {-1}\{1,\dots,k\}$ which we obtain from $G$ if only the $k$ first players (according to $\sigma$) play. 
Let $\Ga^k := \Ga^k(\Id_{[n]})$, where $\Id_{[n]}$ is the identity order (i.e., $\Id_{[n]}(i) = i$). 
For notational convenience, for $k>n$ we set $\Ga^k(\sigma):=\Ga^n(\sigma)=\Ga$. Further, if the order $\sigma$ is defined on a larger domain than the player set of $G$, we define $\Ga^k(\sigma):=\Ga^k(\tau)$ for $\tau:\N\to[n]$ the unique bijection satisfying $\tau(i)<\tau(j)$ iff $\sigma(i)<\sigma(j)$ for all $i,j\in \N$ 


\begin{definition}
Let $\Ga$ be an $n$-player congestion game and let $k\in [n]$. 
An action profile $A$ is a \textbf{$k$-lookahead outcome} of $\Ga$ if there exists an order $\sigma$ on the players so that for each $i\in \N$ we have that $A_{i}$ equals the action $B_i$ played by player $i$ in some subgame-perfect outcome $B$ of $(\Ga')^{k}(\sigma)$ that corresponds to the order $\sigma$, where $\Ga'$ is the subgame of $\Ga$  induced by $(A_j)_{\sigma(j)<\sigma(i)}$.\footnote{Note that we would need to write $(\Ga')^{\min\{k,n-\sigma(i)+1\}}$ instead of $(\Ga')^{k}$ without our assumption that $\Ga^\ell=\Ga^n$ for $\ell>n$.}
\end{definition}

We say a $k$-lookahead outcome corresponds to the order $\tau$ if $\tau$ can be used as the order $\sigma$ in the definition above. We also define a $k$-lookahead outcome for $k>n$ as an $n$-lookahead outcome. 
We use $k\text{-}\mathbb{LO}(\Ga)$ to denote the set of all $k$-lookahead outcomes of a game $\Ga$. 

Assuming $\sigma=\Id_{[n]}$ for ease of notation, $A$ is a $k$-lookahead outcome (corresponding to the identity) if and only if $A_1$ is the action played by the first player in a subgame-perfect outcome (corresponding to the identity) of $\Ga^k$ and $A_{-1}$ is a $k$-lookahead outcome (corresponding to the identity) in the game induced by $(A_1)$. 

\begin{example}\label{ex:1}
Let $\Ga$ be a congestion game with $R=\{r,s,t\}$, $\A_1=\{r,s\}$, $\A_2=\{s,t\}$ and $\A_3=\{t\}$. Let the delay functions be given by
$d_r(x)=2x, ~d_s(x)=1.5x, ~d_t(x)=2x.$
Suppose the players enter the game in the order $1,2,3$ and all anticipate their own decision and the next player ($k=2$). Player 1 then computes the unique subgame-perfect outcome depicted on the left in Figure~\ref{fig:1}. 
Hence he chooses the resource $s$. This brings player 2 in the subgame whose game tree is depicted on the right in Figure~\ref{fig:1}. His unique subgame-perfect choice is $s$. This shows that the only 2-lookahead outcome corresponding to the identity of this game is $(s,s,t)$. 
\end{example}

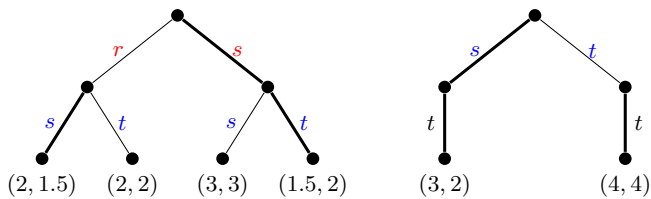
\begin{figure}[t]
\begin{center}
\scalebox{0.95}{\small
\tikzstyle{level 1}=[level distance=1cm, sibling distance=2.5cm]
\tikzstyle{level 2}=[level distance=1cm, sibling distance=1.25cm]
\tikzstyle{bag} = [circle, minimum width=5pt,fill, inner sep=0pt]
\tikzstyle{end} = [circle, minimum width=5pt,fill, inner sep=0pt]
\hspace*{-.3cm}
\begin{tikzpicture}[]
\node[bag] {}
    child {
        node[bag] {}        
            child {
                node[end, black,label=below:
                    {$(2,1.5)$}] {}
                edge from parent[very thick]
                node[left] {$\columncolor{s}$}
            }
            child {
                node[end, black, label=below:
                    {$(2,2)$}] {}
                edge from parent[thin]
                node[right] {$\columncolor{t}$}
            }
            edge from parent[thin] 
            node[left] {$\rowcolor{r}$}
    }
    child {
        node[bag] {}        
        child {
                node[end, label=below:
                    {$(3,3)$}] {}
                edge from parent[thin]
                node[left] {$\columncolor{s}$}
            }
            child {
                node[end, label=below:
                    {$(1.5,2)$}] {}
                edge from parent[very thick]
                node[right] {$\columncolor{t}$}
            }
        edge from parent[very thick]         
            node[right] {$\rowcolor{s}$}
    };
\end{tikzpicture} \qquad
\begin{tikzpicture}[]
\node[bag] {}
    child {
        node[bag] {}        
            child {
                node[end, black,label=below:
                    {$(3,2)$}] {}
                edge from parent[very thick]
                node[left] {${t}$}
            }
            edge from parent[very thick] 
            node[left] {$\columncolor{s}$}
    }
    child {
        node[bag] {}        
            child {
                node[end, label=below:
                    {$(4,4)$}] {}
                edge from parent[very thick]
                node[right] {${t}$}
            }
        edge from parent[thin]         
            node[right] {$\columncolor{t}$}
    };
\end{tikzpicture}
}
\end{center}
\vspace*{-.2cm}
\caption{Illustration of game trees for Example~\ref{ex:1}.}\label{fig:1}
\end{figure}

\paragraph{Inefficiency of lookahead outcomes.}

We introduce the \emph{$k$-Lookahead Price of Anarchy ($k$-\LPoA)} to study the efficiency of $k$-lookahead outcomes, which generalizes both the Price of Anarchy \cite{kouts} and the Sequential Price of Anarchy \cite{leme}. 
We consider two social cost functions in this paper: Given an outcome $A$, the \emph{utilitarian social cost} is defined as $SC(A)=\sum_{i\in \N}c_i(A)$; the \emph{egalitarian social cost} is given by $W_A=\max_{i\in \N}c_i(A)$. 

Given a game $\Ga$, let $A^*$ refer to an outcome of minimum social cost. 
The \textit{$k$-Lookahead Price of Anarchy ($k$-\LPoA)} of a congestion game $\Ga$ is 
\begin{equation}\label{eq:LPoA}
k\text{-}\LPoA(\Ga) = \max_{A\in k\text{-}\mathbb{LO}(\Ga)} \frac{\sum_{i\in \N}c_i(A)}{\sum_{i\in \N}c_i(A^*)}.
\end{equation} 
Here the $k$-\LPoA\ is defined for the utilitarian social cost; it is defined analogously for the egalitarian social cost.
Recall that the \emph{Price of Anarchy (\PoA)} and the \emph{Sequential Price of Anarchy (\SPoA)} refer to the same ratio as in \eqref{eq:LPoA} but replacing ``$k\text{-}\mathbb{LO}(\Ga)$'' by ``$\NE(\Ga)$'' and ``$\SPO(\Ga)$'', respectively. The \emph{Price of Stability (\PoS)} refers to the same ratio, but minimizing over the set of all Nash equilibria.

\paragraph{Curse of ties.}

When studying sequential-move versions of games, results can be quite different depending on whether or not players have to resolve ties. We next introduce two notions to avoid/resolve ties. 

We first introduce the notion of a \emph{generic} congestion game. 
Intuitively, this means that every player has a unique preference among all available actions. 

\begin{definition}
\label{def:genericgame}
A \Rosenthal game $\Ga$
is \textbf{generic} if for all $N\subseteq \N$, $A,B\in\prod_{i\in N} \A_i$ and $j\in N$, $A_j\neq B_j$ implies $c_j(A) \neq c_j(B)$. 
\end{definition}

Note that if $\Ga$ is generic then $\Ga^k(\sigma)$ is also generic for every order of the players $\sigma$ and $k\in \mathbb{N}$ and every induced subgame is generic as well.
Hence if $A$ is the unique subgame-perfect outcome (say respect to the identity) of $\Ga$, then the subgame $G'$ induced by $(A_1)$ is again generic, so that $(A_2,\dots,A_n)$ must be the \emph{only} subgame-perfect outcome (with respect to the identity) of $G'$. With induction, it then follows that for generic games, there is a unique $n$-lookahead outcome which is equal to the unique subgame-perfect outcome. For non-generic games, each subgame-perfect outcome is an $n$-lookahead outcome, but the reverse may be false. 

To see this, consider a congestion game with $R=\{r,s,t\}$, $\A_1=\{r,s\}$ and $\A_2=\{s,t\}$. Let $d_r(x_r)=2x_r$, $d_s(x_s)=2x_s$ and $d_t(x_t)=4$. Then $(s,t)$ is a subgame-perfect outcome, correspondig to the SPE depicted to the left below, so player 1 may choose $s$ in an $n$-lookahead outcome. 

\bigskip
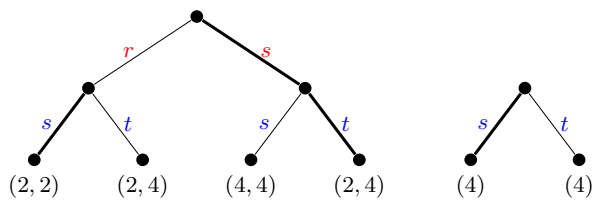
\begin{center}
\scalebox{0.95}{\small
\tikzstyle{level 1}=[level distance=1cm, sibling distance=3cm]
\tikzstyle{level 2}=[level distance=1cm, sibling distance=1.5cm]
\tikzstyle{bag} = [circle, minimum width=5pt,fill, inner sep=0pt]
\tikzstyle{end} = [circle, minimum width=5pt,fill, inner sep=0pt]
\begin{tikzpicture}[]
\node[bag] {}
    child {
        node[bag] {}        
            child {
                node[end, black,label=below:
                    {$(2,2)$}] {}
                edge from parent[very thick]
                node[left] {$\columncolor{s}$}
            }
            child {
                node[end, black, label=below:
                    {$(2,4)$}] {}
                edge from parent[thin]
                node[right] {$\columncolor{t}$}
            }
            edge from parent[thin] 
            node[left] {$\rowcolor{r}$}
    }
    child {
        node[bag] {}        
        child {
                node[end, label=below:
                    {$(4,4)$}] {}
                edge from parent[thin]
                node[left] {$\columncolor{s}$}
            }
            child {
                node[end, label=below:
                    {$(2,4)$}] {}
                edge from parent[very thick]
                node[right] {$\columncolor{t}$}
            }
        edge from parent[very thick]         
            node[right] {$\rowcolor{s}$}
    };
\end{tikzpicture}\qquad
\tikzstyle{level 1}=[level distance=1cm, sibling distance=1.5cm]
\begin{tikzpicture}[]
\node[bag] {}
            child {
                node[end, black,label=below:
                    {$(4)$}] {}
                edge from parent[very thick]
                node[left] {$\columncolor{s}$}
            }
            child {
                node[end, black, label=below:
                    {$(4)$}] {}
                edge from parent[thin]
                node[right] {$\columncolor{t}$}
            };
\end{tikzpicture}
}
\end{center}

However, in the subgame induced by $(s)$, whose game tree is depicted to the right, both $s$ and $t$ are subgame-perfect (player 2 is indifferent), so player 2 may respond $s$, giving the $n$-lookahead outcome $(s,s)$. This is not a subgame-perfect outcome. 

Rather than assuming that no ties exist, we can also restrict the definition of a subgame-perfect outcome:
A \textit{tie-breaking rule} for player $i$ is a total partial order $\succ_i$ on the action set $\A_i$. If player $i$ adopts the tie-breaking rule $\succ_i$, then $A_i$ only is a best response to $A_{-i}$ if
$c_i(A_i,A_{-i})<c_i(B,A_{-i}) \text{ for all } B\in \A_i \text{ such that } B\succ_i A.$
In particular, $\succ_i$ ensures that player $i$ has a unique best response. As a result, there is a unique subgame-perfect outcome for each tie-breaking rule and order of the players. For symmetric congestion games, we can consider the special case of a \textit{common tie-breaking rule} $\succ$ on $\A$ that all players adopt.

\paragraph{Effect of the player order.} Whether all subgame-perfect outcomes are stable may depend on the order of the players. For example, for consensus games all subgame-perfect outcomes corresponding to a \emph{tree respecting order} are stable (see Example \ref{ex:treeorder}). Whenever we make a claim such as ``all subgame-perfect outcomes are stable'', this should be read as ``subgame-perfect outcomes are stable \emph{for all orders}''. 

\begin{theorem}
If $\Ga$ is a symmetric congestion game and $A$ a $k$-lookahead outcome with respect to $\Id_{[n]}$, then $(A_{\sigma(1)},\dots, A_{\sigma(n)})$ is a $k$-lookahead outcome with respect to $\sigma$ for every order $\sigma$.
\end{theorem}
\begin{proof}
Let $A$ be a $k$-lookahead outcome with respect to $\Id_{[n]}$. Then
$A_1$ is the first action of some SPO $B$ with respect to the order $\Id_{[n]}$ in $\Ga^{k}(\Id)$. Let $S$ be a subgame-perfect equilibrium inducing the SPO $B$. 

In the sequential move-version $(\Ga^k(\sigma))^{\sigma}$, the root node is a decision node for player $i=\sigma^{-1}(1)$ and $S_1$ is a valid strategy for player $i$ (since the game is symmetric). Similarly, $S_{2},\dots,S_k$ are valid strategies for $\sigma^{-1}(2),\dots,\sigma^{-1}(k)$. In fact, the game tree of $(\Ga^k)^{\Id}$ is the same as the game tree of $(\Ga^k(\sigma))^{\sigma}$ up to a relabelling of the players. This means that in both games we verify the same equations when determining whether $S$ is a subgame-perfect equilibrium. Hence $S$ is a subgame-perfect equilibrium in $(\Ga^k(\sigma))^{\sigma}$ as well. In the corresponding subgame-perfect outcome of $\Ga^k(\sigma)$, player $i=\sigma^{-1}(1)$ plays $B_1=A_1=A_{\sigma(i)}$. 
Similar argumentation shows that in the subgame $\Ga'$ induced by $(A_1)$, the action $A_2$ is the action of $\sigma^{-1}(2)$ in a subgame-perfect outcome of $(\Ga')^k(\sigma)$.
\end{proof}
This theorem allows us to assume that the order is the identity when proving stability for symmetric games. Moreover, any permutation of a $k$-lookahead outcome is again a $k$-lookahead outcome, which is a useful fact that we exploit in the proofs below.

\medskip
\noindent
Due to lack of space, some of the proofs are omitted from the main text below and will be provided in the full version of the paper.

\section{Symmetric network congestion games}
\label{sec:sncg}

A \textit{single-commodity network} is a directed multigraph $\Gr=(V,E)$ with two special vertices $o,d\in V$ such that each arc $a\in E$ is on at least one directed $o,d$-path.
In a \textit{symmetric network congestion game (SNCG)}, the common set of actions $\A$ is given by the set of all directed paths in a single-commodity network $\Gamma=(V,R)$. 
A \textit{series-parallel graph (SP-graph)} either consists of (i) a single arc, or (ii) two series-parallel graphs in parallel or series. 
An \textit{extension-parallel graph (EP-graph)} either consists of (i) a single arc, (ii) two extension-parallel graphs in parallel, or (iii) a single arc in series with an extension-parallel graph.

In our proofs we exploit the following equivalences:
\begin{lemma}[Nested intersections property]
The common action set $\A$ of each SNCG on an EP-graph satisfies the following three equivalent properties:
\begin{enumerate}

\item For all distinct $A,B,C \in \mathcal{A}$ either $A\cap B\subseteq A\cap C$ or $A \cap B \supseteq A \cap C$. 

\item For all distinct $A,B,C \in \mathcal{A}$, 
$A\cap B \not \subset C$ implies $A \cap C = B\cap C = A\cap B\cap C$. 

\item There is no \emph{bad configuration} (see \cite{strongeq}), i.e., for all distinct $A,B,C \in \mathcal{A}$, $A \cap \left( C \setminus B\right) =\emptyset$ or $A \cap \left( B \setminus C \right) =\emptyset$. 
\end{enumerate}
\end{lemma}
The non-trivial part that an SNCG has no bad configuration if and only if its network is extension-parallel is shown by Milchtaich \cite{milnetw}.

Fotakis et al. \cite{fot06} show that each 1-lookahead outcome is a Nash equilibrium for SNCG on SP-graphs.
We prove that the converse also holds for EP-graphs. 
\begin{theorem}
\label{thm:NEisBRP}
For every SNCG on an EP-graph, the set of 1-lookahead outcomes coincides with the set of Nash equilibria.
\end{theorem} 
\begin{proof}
It remains to show that each Nash equilibrium of $\Ga$ is a permutation of a 1-lookahead outcome corresponding to the identity.

Consider an NE $A$ with corresponding congestion vector $x$. Let $S_1$ be the set of actions costing the least for the first player, that is, those minimising $\sum_{r\in P}d_r(1)$. This corresponds to the set of NE of the 1-player game. Assume towards contradiction no one plays an action from $S_1$. Let $P\in S_1$. Since $A$ is an NE, $P$ has become more expensive as more players joined in, so some $A_j$ overlaps with it. Pick $j$ with $|P\cap A_j|$ maximal. Then by the nested intersection property, no $A_k$ intersects with $P\setminus A_j$. Hence $x_r=0$ for any $r\in P\setminus A_j$ (these are not chosen). We find
$$
\sum_{r\in P\setminus A_j }d_r(x_r+1)=\sum_{r\in P\setminus A_j }d_r(1)<\sum_{r\in A_j \setminus P }d_r(1) \leq \sum_{r\in A_j \setminus P }d_r(x_r)
$$
using that $A_j\not\in S_1$ and $P\in S_1$. This contradicts the fact that $A$ is an NE. 

Let $j$ be a player picking an action from $S_1$, so that $(A_j)$ forms a 1-lookahead outcome. Define $\sigma(j)=1$. Suppose we have defined $\sigma$ so that $A'= (A_{\sigma^{-1}(1)},\dots,A_{\sigma^{-1}(i)})$ forms a 1-lookahead outcome for some $1\leq i<n$. The profile 
$$
A''=A\setminus A'= (A_{j})_{j\in \sigma^{-1}\{i+1,\dots,n\}}
$$
forms an NE in the game $\Ga'$ induced by $A'$. Repeat the argument above for $A''$ and $\Ga'$ to define $\sigma^{-1}(i+1)$. 
\end{proof}

\subsection{Stability and inefficiency of generic games}

As shown in the introduction, SPOs are not guaranteed to be stable for SNCGs on EP-graphs. However, stability is guaranteed if the game is generic. 
\begin{theorem}
\label{thm:spoep}
For every generic SNCG on an EP-graph, the set of subgame-perfect outcomes coincides with the set of Nash equilibria.
\end{theorem}
\begin{proof}
Because the game is generic, there is a unique 1-lookahead outcome (up to permutation) and hence a unique NE: The game $\Ga^1$ is generic and thus there is a unique cheapest first path $B_1$. The game $\Ga^2$ is generic and hence so is the subgame of $\Ga^2$ induced by $(B_1)$. Thus there is a unique cheapest second path, and so on. 

We prove the statement by induction on the number of players $n$. The claim is true for $n=1$. Suppose the claim holds for all generic SNCGs on EP-graphs with less than $n$ players and let $\Ga$ be an SNCG on an EP-graph with $n$ players. 
Let $A$ be an SPO and $S$ the corresponding SPE, corresponding to some order $\sigma$ which we may assume to be the identity by relabeling the players. 
Given an action $P$ of player 1, $S$ prescribes an SPE $S'$ in the game induced by $(P)$. By our induction hypothesis (the subgame induced by $(P)$ has $n-1$ players), the SPO corresponding to $S'$ is a permutation of a 1-lookahead outcome in this subgame. In particular, if $B$ is the unique 1-lookahead outcome of the game $\Ga$ and player 1 plays $B_i$, then the resulting outcome according to $S$ is a permutation of $B$. 

Suppose towards contradiction that  $A_1\neq B_i$ for all $i$.
By induction $A_{-1}$ is a NE in the subgame induced by $(A_1)$ (by the same argument as before). Since $A$ is not a Nash equilibrium and all players except 1 are playing a best response, there is an action $P$ so that
$c_1(P,A_{-1})<c_1(A)$.
If player 1 switches to $P$, then the other players are still playing a best response by 
\cite[Lemma 1]{fotakis}. So $(P,A_{-1})$ is an NE, hence a permutation of $B$. This means that $P=B_i$ for some $i$ (and $A_{-1}$ some permutation of $B_{-i}$), which yields a contradiction: $c_1(B_i,B_{-i})=c_1(P,A_{-1})<c_1(A_1,A_{-1})\leq c_1(B_i,B_{-i})$, where the last inequality follows because $A_1$ is subgame-perfect for player 1.
\end{proof}

\begin{theorem}
\label{thm:corNE}
For every SNCG on an EP-graph, each Nash equilibrium is a subgame-perfect outcome. 
\end{theorem}
\begin{proof}[Sketch]
Call game $\Ga'$ (with delay functions $d'_r$) \textit{close} to $\Ga$ if for all paths $P\neq Q$ and congestion vectors $(x_r)$ and $(y_r)$
$$
\sum_{r\in P} d_r(x_r) < \sum_{r\in Q} d_r(x_r) \implies \sum_{r\in P} d'_r(x_r) \leq \sum_{r\in Q} d'_r(x_r).
$$
Suppose $\Ga$ is given together with a Nash equilibrium $A$. 
If $\Ga'$ is a generic game close to $\Ga$ which also has $A$ has Nash equilibrium, then by Theorem \ref{thm:spoep} this is the unique subgame-perfect outcome of $\Ga'$. The close condition above (invented by Milchtaich \cite{crowdgames}) exactly ensures $A$ is then also a subgame-perfect outcome in $\Ga$. It remains to find a close generic game, which can be done by adjusting the delay functions of $\Ga$; each path contains a resource he shares with no other path \cite{milnetw}, so we can increase the cost of individual paths without effecting the other path costs. 
\end{proof}

The following theorem is the main result of this section. 

\begin{theorem}
\label{thm:kLOforSNCG}
Let $G$ be a generic SNCG on an EP-graph. Then for every $k$ the set of $k$-lookahead outcomes coincides with the set of Nash equilibria.
As a consequence, $k$-$\LPoA(\Ga)=\PoA(\Ga)$. 
\end{theorem}
\begin{proof}
Since there is a unique Nash equilibrium and $k$-lookahead outcome up to permutation, it suffices to show that each $k$-lookahead outcome is a Nash equilibrium. By relabelling the players, we can assume that the order of the players is the identity.
Let $B$ denote the unique 1-lookahead outcome of $\Ga$ corresponding to the identity. Then $\Ga^k$ has unique 1-lookahead outcome $(B_1,\dots,B_k)$.

Let $A$ be a $k$-lookahead outcome. Since each SPO of $\Ga^k$ is a permutation of a 1-lookahead outcome, we know that $A_1\in \{B_1,\dots,B_k\}$. This implies that the subgame of $\Ga^{k+1}$ induced by $(A_1)$ has some permutation of $(B_1,\dots,B_{k+1})\setminus (A_1)$ as unique 1-lookahead outcome (where we perform the set operations seeing the tuple as a multiset). This means that $A_2\in (B_1,\dots,B_{k+1})\setminus (A_1)$. Continuing this way, we see that
$A_i \in (B_1,\dots,B_{k+i-1})\setminus (A_1,\dots,A_{i-1})$ 
for $i=1,\dots, n-k$ and for $i>n-k$ we find
$A_i \in B\setminus (A_1,\dots,A_{i-1})$. Hence $A$ will be a permutation of $B$.
\end{proof}

The result of Theorem~\ref{thm:kLOforSNCG} does not extend to series-parallel graphs.
\begin{proposition}
\label{prop:brpnotspo}
For any SP-graph $\Gr$ that is not EP, there is a generic SNCG $\Ga$ on $\Gr$ such that the sets of 1-lookahead and $n$-lookahead outcomes are disjoint.
\end{proposition}
\begin{proof}
Consider the single-commodity network with four arcs $r,s,t,u$ and $od$-paths $rt$, $st$, $ru$ and $su$. Each series-parallel graph that is not extension-parallel has this network as a minor (see \cite{milnetw}). Thus it suffices to give a counterexample on this network. 
Consider the generic three-player game with delay functions 
\[
\begin{array}{l}
(d_r(1),d_r(2),d_r(3)) = (1,3,100)\\ (d_s(1),d_s(2),d_s(3)) = (2,4,200) \\
(d_t(1),d_t(2),d_t(3)) = (1.1,4.1,100.1)\\ (d_u(1),d_u(2),d_u(3)) = (2.2,3.2,100.2).
\end{array}
\]
The unique 1-lookahead outcome (up to permutation) is $(rt,su,ru)$ and the unique $n$-lookahead outcome (up to permutation) is $(st,ru,ru)$. 
\end{proof}

We next show that anticipation may still be beneficial for the first player.
\begin{theorem}
Let $\Ga$ be a generic SNCG on an EP-graph. Let $B$ be a subgame-perfect outcome with respect to the identity. Then $c_1(B)\leq\dots\leq c_n(B)$. In particular, $c_1(B)\leq c_1(A)$ for any $k$-lookahead outcome $A$ with respect to the identity. 
\end{theorem}
\begin{proof}
Let $A=(A_1,\dots,A_n)$ be a permutation of the unique NE of $\Ga$ for which $c_1(A)\leq \dots \leq c_n(A)$. Let $B$ be the unique subgame-perfect outcome with respect to the identity; then this is some permutation of $A$ by Theorem~\ref{thm:kLOforSNCG}.

If player 1 plays $A_1$, then his successors will play $A_2,\dots,A_n$ (not necessarily in that order), so that $c_1(B)\leq c_1(A)$. Since $B$ is some permutation of $A$ and since $A_1$ (which may equal $A_i$ for some $i$) is the unique element from $\A$ with the lowest cost in the profile $A$, we find $A_1=B_1$ and $c_1(A)=c_1(B)$.
Similarly, player $j$ can ensure himself the cost $c_j(A)$ in the subgame induced by $(A_1,\dots,A_{j-1})$ and therefore $c_j(B)=c_j(A)$. This proves $c_1(B)=c_1(A)\leq\dots\leq c_n(A)=c_n(B)$

Finally, note that if $C$ is a $k$-lookahead outcome, then $C$ is some permutation of $A$ and therefore $c_1(C)=c_j(A)=c_j(B)\geq c_1(B)$ for some $j\in \N$.
\end{proof}


The following example shows that the cost of the first players does not decrease monotonically with his lookahead. In fact, a generalization of this example shows that only full lookahead guarantees the smallest cost for player 1. 

\begin{example}
Consider the generic symmetric singleton congestion game with $R=\{r,s\}$, $d_r(x)=x$ and $d_s(x)=x+0.5$. The subgame-perfect outcome with respect to the identity is $(r,r,\dots r, s,\dots,s)$ if $n$ is even and $(s,s,\dots,s,r,\dots,r)$ if $n$ is odd. Thus, if $A$ is a $(n-2)$-lookahead outcome and $B$ a $(n-1)$-lookahead outcome, then $c_1(A)=\min_{i\in \N}c_i(A)<\max_{i\in \N}c_i(A)=c_1(B)$. 
\end{example}

In contrast to the above, the first player is not guaranteed to achieve minimum cost with full lookahead if the game is non-generic. 

\begin{example}
Suppose a symmetric singleton congestion game with delay functions $d_r(x)=x=d_s(x)$ is played by an odd number of players. The successors of player 1 can decide which resource becomes the most expensive one and enforce that this is the one that player 1 picks, so that player 1 is always worse off.
\end{example}

\subsection{Inefficiency of non-generic games}


Let $\Ga$ be an SNCG and let $A$ be an outcome. Let $(x_r)_{r\in R}$ be the corresponding congestion vector. We define the \emph{opportunity cost} of $A$ in $\Ga$ as the minimal cost that a new player entering the game would have to pay, i.e., $O_A(\Ga)=\min_{P\in \A}\sum_{r\in P}d_r(x_r+1)$; this definition is implicit in \cite{optimalNEmakespan,optimalNE}.
The \textit{worst cost} of $A$ in $\Ga$ is the egalitarian social cost, i.e., $W_A(\Ga)= \max_{i\in\N}\sum_{r\in A_i}d_r(x_r)$. 
\begin{lemma}
\label{lem:oppcosts}
Let $\Ga$ be an SNCG on a SP-graph $\Gr$ with $m$ players and let $\Ga^n$ be the same game with $n\leq m$ players. 
For any action profile $A$ in $\Ga^n$ and any 1-lookahead outcome $B$ of $\Ga$, we have $O_A(\Ga^n)\leq O_B(\Ga)$. Further, if $n<m$, then $O_A(\Ga^n)\leq W_B(\Ga)$.
\end{lemma}
The lemma above can be proved by induction on the graph structure.
\begin{corollary}
\label{cor:brpsamepotential}
For SNCGs on SP-graphs, all 1-lookahead outcomes have the same value for the potential function. 
\end{corollary}
\begin{proof}
Let $A=(A_1,\dots,A_n)$ and $B=(B_1,\dots,B_n)$ be 1-lookahead outcomes (with respect to the identity). If $n=1$, then since both are chosen greedily we must have $\Phi(A)=\Phi(B)$. For $n>1$, the profiles $A'= (A_1,\dots,A_{n-1})$ and $B'=(B_1,\dots,B_{n-1})$ are also 1-lookahead outcomes, so that by an inductional argument and using that 1-lookahead outcomes have the same opportunity costs (Lemma \ref{lem:oppcosts}), we find
$$
\Phi(A) = c_n(A) + \Phi(A')  = O_{A'} +\Phi(A')= O_{B'}+\Phi(B')= \Phi(B).
$$
\end{proof}
Applying induction on the graph structure again, we are able to derive that each 1-lookahead outcome is a global optimum of Rosenthal's potential function.
\begin{proposition}
\label{prop:optima}
Let $\Ga$ be an SNCG on a series-parallel graph. Let $\mathcal{M}$ denote the set of global minima of the Rosenthal potential function $\Phi(A)=\sum_{r\in R}\sum_{i=1}^{x(A)_r}d_r(i)$. 
\begin{enumerate}
\item For any $A\in \mathcal{M}$, there is a 1-lookahead outcome $B$ with $x(B)=x(A)$.
\item All 1-lookahead outcomes of $\Ga$ are global minima of the potential function and
$$
\{x(B)\mid B \text{ 1-lookahead outcome}\}=\{x(A)\mid A \in \mathcal{M}\}.
$$
\end{enumerate}
\end{proposition}
Fotakis \cite[Lemma 3]{fotakis} proves the following proposition:  

\begin{proposition}[Lemma 3 in \cite{fotakis}]
\label{thm:fotakis}
Let $\Ga$ be an SNCG on an SP-graph with delay functions in class $\mathcal{D}$. For any global minimum $A$ of the Rosenthal potential function and any other profile $B$ we have $\sum_{i\in \N}c_i(A) \leq \rho(\mathcal{D}) \sum_{i\in \N}c_i(B)$,
where 
$$
\rho(\mathcal{D})=\left(1-\sup_{d\in \mathcal{D}\setminus \{0\}}\sup_{x,y\geq 0} \frac{y(d(x)-d(y))}{xd(x)}\right)^{-1}.
$$
\end{proposition}
For example, if $\mathcal{D}$ consists of all affine functions then $\rho(\mathcal{D}) = \frac{4}{3}$.
Fotakis \cite{fotakis} also proves that for SNCGs on SP-graphs, the Price of Stability is $\rho(\mathcal{D})$. Combining all these results, we obtain the following corollary.
\begin{corollary}
\label{cor:brps}
For every SNCG on an SP-graph with delay functions in class $\mathcal{D}$, we have $1$-$\LPoA \leq \rho(\mathcal{D})=\PoS$.
\end{corollary}
Note that 1-lookahead outcomes are only guaranteed to be optimal Nash equilibria for games with the worst Price of Stability; a procedure for finding an optimal Nash equilibrium for every SNCG on a SP-graph is an NP-hard problem \cite{optimalNE}.

For series-parallel graphs, 1-lookahead outcomes have maximal worst cost among the Nash equilibria \cite{optimalNEmakespan}; we use $W(\Ga)$ to refer to the common worst cost of 1-lookahead outcomes.
\begin{lemma}
\label{lem:sncgworstcost}
Let $\Ga$ be an SNCG on an SP-graph with $n$ players and common action set $\A$. Let $P_1,\dots,P_m\in \A$ for $m< n$ and let $\Ga'$ be the subgame of $\Ga$ induced by $(P_1,\dots,P_m)$. Then $W(\Ga')\leq W(\Ga)$. 
\end{lemma}
\begin{proof}
The subgame $\Ga'$ is again an SNCG (with $n-m$ players). By Lemma 1 in \cite{optimalNEmakespan}, the last player pays the worst cost in any 1-lookahead outcome $B$ of $\Ga'$. Moreover, the last player pays at most $W(\Ga)$:  by Lemma \ref{lem:oppcosts}, the profile $A'=(P_1,\dots,P_m, B_1,\dots,B_{n-m-1})$ satisfies $O_{A'}(\Ga^{n-1})\leq W(\Ga)$ where $\Ga^{n-1}$ denotes the game $\Ga$ with $n-1$ players. Hence the greedy choice will cost at most $W(\Ga)$.   
\end{proof}

In particular, when at most $n-1$ players have chosen a path, there is still a path available of cost at most $W(\Ga)$. On EP-graphs, Nash equilibria (and thus also 1-lookahead outcomes) have optimal egalitarian social cost \cite{epstein}.
\begin{theorem}
\label{thm:SPOgoodcost}
For every SNCG $\Ga$ on an EP-graph, each SPO $A$ has optimal egalitarian social cost, i.e., $W_A=W(\Ga)$.
\end{theorem}
\begin{proof}
We prove the claim by induction on the number of players in the game. If $\Ga$ has $n=1$ players then the claim is true. Assume the claim is true for all subgame-perfect outcomes $A$ in games $\Ga$ with at most $k$ players for some $k\geq 1$. Suppose towards contradiction that there is a game $\Ga$ with $n=k+1$ players for which there exists an SPO $A$ with $W_A(G)>W(\Ga)$. Renumber the players so that $A$ corresponds to the identity.
Let $j$ be the last player paying a cost worse than $W(\Ga)$ in the profile $A$. If $j\neq 1$, then the subgame $\Ga'$ induced by $(A_1,\dots,A_{j-1})$ has at most $k$ players so that by induction the SPO $A'=(A_j,\dots,A_n)$ satisfies $W(\Ga')=W_{A'}(\Ga')>W(\Ga)$ (the strict inequality follows by the choice of $j$). This is a contradiction with Lemma \ref{lem:sncgworstcost}. Hence we assume $j=1$. 

Let $P$ minimize $\sum_{r\in P}d_r(1)$, i.e., choose the cheapest path at the start of the game. 
Consider the subgame $\Ga'$ induced by $(P)$. Take a subgame-perfect outcome $B=(B_i)_{i>1}$ for this subgame. Let $x$ be the congestion vector corresponding to $(P,B_{2},\dots,B_n)$.
\begin{enumerate}
\item Since player $1$'s subgame-perfect move costs more than $W(\Ga)$, path $P$ will eventually cost $c(P,x) :=\sum_{r\in P}d_r(x_r)> W(\Ga)$.
\item Since $B$ is an SPO in $\Ga'$ and $\Ga'$ has strictly less players than $\Ga$, we get $W_{B}(\Ga')=W(\Ga')$ by induction. Hence $c(B_i,x)\leq W_B(\Ga')=W(\Ga')\leq W(\Ga)$ for all $i>1$, where the last inequality follows from Lemma \ref{lem:sncgworstcost}.
\end{enumerate}
By 1. and 2., no successor of player $1$ can pick $P$. By 1., there must be a player picking a path having overlap with $P$. Let $i>1$ for which $|P\cap B_i|$ is largest. By the nested intersection property, no player $k>1$ picks a path overlapping with $P\setminus B_i$. It follows that 
\begin{align*}
c(B_i,x) & \geq \sum_{r\in B_i\cap P}d_r(x_r)+ \sum_{r\in B_i\setminus P} d_r(1) \\
& \geq \sum_{r\in B_i\cap P}d_r(x_r)+ \sum_{r\in P \setminus B_i} d_r(1)=c(P,x),
\end{align*}
where for the second inequality we use that $P$ was the cheapest path at the beginning of the game. By using 1. and 2., we conclude that 
$c(P,x)> W(\Ga)\geq c(B_i,x) \geq c(P,x)$,
which is a contradiction.
\end{proof}



\section{Extensions}
\label{sec:extensions}

We present our results for cost-sharing games and consensus games. 

\subsection{Cost-sharing games} 

A \textit{cost-sharing game} is a congestion game, where the delay functions are non-increasing. 
We first argue that subgame-perfect outcomes are not guaranteed to be stable, even for symmetric singleton cost-sharing games. 

\begin{example}
Consider a cost-sharing game with two resources $r,s$ with identical delay functions 
$d_r(x) = d_s(x) = d(x)$ such that $d(x) = 2$ for $x = 1$ and $d(x) = 1$ for $x > 1$. If three players play this game, the subgame-perfect outcomes $(r,s,s)$ and $(s,r,r)$ (which can be created by defining $S_3(x,y)=y$ and $S_2(x)\neq x$) are unstable. 
\end{example}

On the other hand, the above instability can be resolved for either symmetric or singleton cost-sharing games if no ties exist. 

\begin{theorem}
\label{thm:gscsg}
Let $G$ be a generic, symmetric cost-sharing game. Each $k$-lookahead outcome is a Nash equilibrium of the form $(P_k,\dots,P_k)$ for 
$P_k = \argmin_{A\in \mathcal{A}}\sum_{r\in A} d_r(k).$
In particular, each subgame-perfect outcome is optimal.
\end{theorem}

The proof of Theorem~\ref{thm:gscsg} relies on the following lemma.

\begin{lemma}
\label{prop:uniqueSPO} 
Let $A$ be an outcome so that for any $B$ with $B_i\neq A_i$ for some $i\in \N$, we have $c_i(A)>c_i(B)$. Then $A$ is the unique subgame-perfect outcome.
\end{lemma}

\begin{proof}[Proof of Theorem~\ref{thm:gscsg}]
We first show that each subgame-perfect outcome is of the form $(P_n,\dots,P_n)$, which is the optimal profile. Because the game is generic, no other action $Q\neq P_n$ can cost the same as $P_n$ does in the profile $(P_n,\dots,P_n)$. Hence for any outcome in which some player $i$ plays a different action than $P_n$, this player is \emph{strictly} worse off than in $(P_n,\dots,P_n)$. As a consequence, $(P_n,\dots,P_n)$ is the only subgame-perfect outcome by Lemma \ref{prop:uniqueSPO}.

Let $k\leq n$ be given. The game $\Ga^k$ has a $P_k$ as its unique optimal outcome.
Since $\Ga^k$ is again a generic, symmetric cost-sharing game, each subgame-perfect outcome is of the form $(P_k,\dots,P_k)$. Thus, any $k$-lookahead outcome has the action $P_k$ as its first action. Let $\Ga'$ be the subgame induced by $(P_k)$, then $\Ga'^k$ is again a generic, symmetric cost-sharing game. Since the delay functions are non-increasing, the profile $P_k$ is again the optimal profile. Hence the second player also picks $P_k$. The proof now follows by continuing inductively.
\end{proof}

We turn to inefficiency of $k$-lookahead outcomes. In general, the $k$-$\LPoA$ can increase with $k$ for generic, symmetric singleton cost-sharing games:

\begin{example}
Define delay functions $d_r(x)=\frac{n}x$ and $d_s(1)=n+1$ with $d_s(x)=2$ otherwise. Each $1$-lookahead equilibrium is optimal, whereas no $k$-lookahead equilibrium is optimal for $1<k<n$. This example can be extended make any subset of $[n-1]$ equal to the set $\{k\in[n-1]\mid k\text{-}\LPoA(\Ga)=1\}$. 
\end{example}

However, we can gain monotonicity by restricting the delay functions.
\begin{corollary}
For every generic, symmetric cost-sharing game $\Ga$ with delay functions of the form $d_r(x)=\frac{a_r}x+b_r$, $k$-$\LPoA(\Ga)$ is non-increasing in $k$.
\end{corollary}
\begin{proof}
Theorem~\ref{thm:gscsg} shows that each $k$-lookahead outcome is of the form $(P_k,\dots,P_k)$ for $P_k=\argmin_{A\in \mathcal{A}}\sum_{r\in A} d_r(k)$. Let $n$ the number of players. Denote $a_{k} =\sum_{r\in P_k}a_r$ and $b_{k}=\sum_{r\in P_k}b_r$, so that $\sum_{r\in P_k}d_r(\ell)=\frac{a_k}\ell+b_k$.

If $k$-$\LPoA(\Ga)>(k-1)$-$\LPoA(\Ga)$, then in particular $P_k\neq P_{k-1}$. This implies $\frac{a_{k}}{k-1}+b_{k}>\frac{a_{k-1}}{k-1}+b_{k-1}$ and $\frac{a_{k}}k+b_{k}<\frac{a_{k-1}}{k}+b_{k-1}$. Subtracting the first from the second and dividing by $\frac1{k-1}-\frac1k>0$ shows $a_k>a_{k-1}$. This implies $a_k(\frac1k-\frac1n)>a_{k-1}(\frac1k -\frac1n)$. Subtracting this from $\frac{a_{k}}k+b_{k}<\frac{a_{k-1}}{k}+b_{k-1}$ gives $\sum_{r\in P_k}d_r(n)<\sum_{r\in P_{k-1}}d_r(n)$.
\end{proof}

\begin{theorem}
\label{thm:triv}
If $G$ is a generic singleton cost-sharing game, then each SPO is an NE. However, $k$-lookahead equilibria are not guaranteed to be Nash equilibria for any $k<n$.
\end{theorem}
\begin{proof}
We mimick the proof for fair cost-sharing games \cite{leme}.
Assume without loss of generality that all resources in $R$ can be chosen by some player. For $r\in R$, let $N_r=\{i\in \N:r\in \A_i\}$ be the set of players that can choose resource $r$. Because the game is generic, 
$$
r^*=\arg\min_{r\in R} d_r(|N_r|) 
$$
satisfies $d_{r^*}(|N_{r^*}|)<d_s(|N_s|)$ for all $s\in R\setminus \{r^*\}$. Hence if all successors and predecessors of $i\in N_{r^*}$ choose $r^*$, then this is the unique best response for player $i$ as well. This implies all player $i\in N_{r^*}$ will choose $r^*$ independent of the order the players arrive in. Removing these players from the game and repeating this at most $n$ times, we find that the resource a player picks does not depend on the order in which the players arrive. Since the last player in a subgame-perfect outcome is always giving a best response and each player can be seen as the last player, it follows that this unique SPO is a Nash equilibrium. 

An example of a 1-lookahead outcome which is not an NE can be created by putting
$d_r(x)=\frac1x,~d_s(x)=\frac4{3x}, ~\A_1=\{r,s\}, ~\A_2=\{s\}.$
This example can be extended to give unstable $k$-lookahead equilibria for any $k<n$ by putting players in between that do not interfere with the two players (e.g., by creating a third resource $t$ and setting $\A_i=\{t\}$ for $i\in \{2,\dots,n-1\}$). 
\end{proof}

\subsection{Consensus games}

In a consensus games, each player is a vertex in a weighted graph $\Gr = (V,E,w)$ and can choose between actions $L$ and $R$. The cost of player $i$ in outcome $A$ is given by the sum of the weights $w_{ij}$ of all incident edges $ij\in E$ for which $A_i\neq A_j$. 
The following example shows that subgame-perfect outcomes may be unstable in general. 

\begin{example}
\label{ex:treeorder}
Consider a graph with vertex set $V=\{1,2,3\}$ and edge set $E=\{\{2,3\},\{3,1\}\}$.
Choose weights $w_{23}> w_{13}>0$. For the order $1,2,3$ on the players, there is only one possibility to let player 3 act subgame-perfectly (since $w_{23}>w_{13}$): $S_3(x,y)=y$. Player 2 is hence indifferent between his two actions. So we can define $S_2(x)\neq x$ (``always choose an action different from player 1''). Player 1 is now indifferent as well 
and we can set $S_1=L$. This defines a subgame-perfect equilibrium for which the corresponding SPO is unstable.

Note that the graph in this example forms a tree. Call an order $\sigma$ \emph{tree respecting} if each player except the first succeeds at least one of his neighbours. Mimicking the proof of Proposition \ref{prop:consensusgames}, one can show that for each subgame-perfect outcome is optimal (i.e. of the form $(L,\dots,L)$ or $(R,\dots,R)$)  for such orderings.
\end{example}
\begin{proposition}
\label{prop:consensusgames}
Let $\Ga$ be a consensus game. If all players adopt a common tie-breaking rule, then all $k$-lookahead outcomes are optimal.
\end{proposition}
\begin{proof}
Assume without loss of generality that $R$ is the preferred action of all players and the order on the players is $1,\dots,n$. (Otherwise, we can relabel the players and actions $R$ and $L$.) 

We first show that each SPO $A$ is of the form $(R,\dots,R)$.
Let $S$ be the subgame-perfect equilibrium corresponding to $A$. We show inductively that $S_i(R,\dots,R)=R$ for all $i\in \{n,\dots,1\}$, so that $A=(R,\dots,R)$ follows. The claim for player $n$ follows since $c_n(R,\dots,R)$ is the optimal possible cost for player $n$ and player $n$ plays $R$ when indifferent. Assuming $S_i(R,\dots,R)=R$ holds for players $n,\dots,k$ for some $k\in \{2,\dots,n\}$, player $k-1$ will get cost $c_{k-1}(R,\dots,R)$ if he plays $R$, which again implies that he must choose $R$. 

Let $A$ be a $k$-lookahead outcome for some $k\in\{1,\dots,n\}$.
The only subgame-perfect outcome corresponding to the order $1,\dots,k$ and common tie-breaking towards $R$ is given by $(R,\dots,R)$, so that $A_1=R$. After players 1 up to $\ell$ have fixed action $R$, the profile $(R,\dots,R)$ increases in value at least as much as any other profile, so that $(R,\dots,R)=R^{\min\{k,n-\ell\}}$ will be the only subgame-perfect outcome in the subgame induced by $(A_1,\dots,A_{\ell})=(R,\dots,R)$ for any $\ell\in \{1,\dots,n-1\}$. Hence $A=(R,\dots,R)$. (We find $A=(L,\dots,L)$ if ties are broken in favor of $L$.)
\end{proof}

\section{Future work}

While the focus in this paper is on congestion games, our notion of $k$-lookahead outcomes naturally extends to arbitrary normal-form games (details will be given in the full version of the paper). It will be interesting to study $k$-lookahead outcomes for other classes of games. Also, an interesting direction is to consider heterogeneous (e.g., player-dependent) anticipation levels. In particular, it would be interesting to further explore the relation between ties and anticipation within this framework.  

Another interesting question for future research is to investigate more closely the relationship between greedy best-response and subgame-perfect outcomes. Note that in Section 3 we were able to use results for $k = 1$ and $k = n$ to infer results for other values of $k$. This might be true for other games as well. Especially, it would be interesting to identify general properties of games for which such an  approach goes through.


\paragraph{Acknowledgements.}
We thank Pieter Kleer for helpful discussions. This work was done while the first author was a research intern at CWI and a student at the University of Amsterdam (UvA) .


\newpage





\newpage

\appendix 

\section{Subgame-perfect outcomes}
\label{sec:seqmove}
We formally define subgame-perfect outcomes and Sequential Price of Anarchy for congestion games. All definitions in this section extend easily to general normal-form games. 

When going through the definitions of sequential-move version and subgame-perfect equilibrium, it might be useful to compare them with the example in the introduction.
For any \Rosenthal game and order on the players, we can draw a game tree. This is formalised below.
\begin{definition}
Let $\Ga$ be an $n$-player \Rosenthal game. Let $\sigma:\N\to[n]$ be an order on the players (i.e. a bijection). The \textbf{sequential-move version} $\Ga^\sigma$ of $\Ga$ is the \emph{extensive-form game} in which each player $i\in \N$ chooses his action after observing the actions played by all players $j\in \N$ before him (those for which $\sigma(j)<\sigma(i)$). We represent the game by a directed tree. Each non-leaf node is a decision state and corresponds to a player and a set of outgoing arcs. Each leaf is an end state and corresponds to a cost vector. 
\begin{itemize}
\item The root node is the only state of the first player $\sigma^{-1}(1)$.
\item For any player $i\in \N$ there are $|\A_i|$ outgoing arcs from any of his states, each of which indexed by an action $A_i\in \A_i$. If $\sigma(i)\neq n$, then the head of an action $A_i\in \A_i$ is a state of his successor $\sigma^{-1}(\sigma(i)+1)$. If $\sigma(i)=n$ then each $A_i$ points to a leaf node.
\item Each leaf node has a cost vector $(c_i(A))_{i\in \N}$ for $A$ the profile on the path from the root to the leaf. (This gives an action profile because the unique path from the leaf to the root contains exactly one action from each player.)
\end{itemize}  
This defines the game tree. The set of players of $\Ga^\sigma$ is $\N$ again. A \emph{strategy} $S_i$ of a player $i\in \N$ in $\Ga^\sigma$ is a function
\[
S_i:\prod_{j\in \N:\sigma(j)<\sigma(i)} \A_j \to \A_i
\]
that maps any state of player $i$ to a single outgoing arc.
\end{definition}
Given strategies $S_i$ for $i\in \N$ in a sequential-move version $\Ga^\sigma$ of a game $\Ga$, we refer to $S=(S_i)_{i\in \N}$ as a \emph{strategy profile}. We correspond the following action profile $A$ in $\Ga$ to this. Let  $A_{\sigma^{-1}(1)}:=S_{\sigma^{-1}(1)}$. Having defined $A_{\sigma^{-1}(1)},\dots,A_{\sigma^{-1}(k)}$ for $1\leq k<n$, we let 
\[
A_{\sigma^{-1}(k+1)}=S_{\sigma^{-1}(k+1)}(A_{\sigma^{-1}(1)},\dots,A_{\sigma^{-1}(k)}).
\]
That is, we color the arcs that are prescribed by the strategies (similar to figure in the introduction) and follow the unique path from the root to a leaf. Let $A(S)$ be the action profile corresponding to $S$ in this way.

Write $S_{-i}=(S_j)_{j\in \N\setminus \{i\}}$. The strategy $S_i$ is a \emph{best response} to $S_{-i}$ if $c_i(A(S))\leq c_i(A(T))$ for all strategy profiles $T$ with $T_{-i}=S_{-i}$. A strategy profile $S$ is a \emph{Nash equilibrium} if $S_{i}$ is a best response to $S_{-i}$ for all $i\in \N$. 

Note that a Nash equilibrium $S$ of a sequential-move version is incomparable to a Nash equilibrium $B$ of the original game. However, we can compare $A(S)$ and $B$ and we will do this with an even stronger requirement on $S$.

We first need to define a subgame, which corresponds to fixing the actions of a subset of the players and looking at the game the remaining players are left with.
\begin{definition}
\label{def:inducedgame}
Let $\Ga=(\N,R,(\A_i)_{i\in \N},(d_r)_{r\in R})$ be a \Rosenthal game. Let $\N'\subseteq \N$ and $B_j\in \A_j$ for $i\in \N'$. The \emph{subgame induced by the action profile} $(B_j)_{j\in \N'}$ is the \Rosenthal game
\[
\Ga'=(\N\setminus \N', R, (\A_i)_{i\in \N\setminus \N'},(d'_r)_{r\in R}),
\]
where 
\[
d'_r:\mathbb{N}\to \R:y\mapsto d_r(y+|\{i\in \N':r\in B_i\}|).\qedhere
\]
\end{definition}
Suppose that in the definition above the players in $\N'$ act first in our order, that is, let $\sigma:\N\to\N$ a bijection so that $\sigma^{-1}\{1,\dots,|\N'|\}=\N'$. We call the sequential-move version $\Ga'^\sigma$ of $\Ga'$ a \emph{subgame} of $\Ga^\sigma$. The game tree of $\Ga'^\sigma$ is a subtree of the game tree of $\Ga^\sigma$, namely the one with as root the decision node of player $\sigma^{-1}(|\N'|+1)$ in which he needs to specify his action after observing $(B_i)_{i\in \N'}$. All such ``subtrees" give subgames of $\Ga^\sigma$.  
\begin{definition}
A \textbf{subgame-perfect equilibrium} (SPE) is a strategy profile $(S_i)_{i\in \N}$ that induces a Nash equilibrium in every subgame of a given sequential-move version. 
\end{definition}
The proposition below is implicit in the work of Zermelo \cite{zermelo} and often seen as one of the first results in game theory.
\begin{proposition}
Every sequential-move version has a subgame-perfect equilibrium.
\end{proposition}
\begin{proof} 
Let $\Ga^\sigma$ be a given sequential-move version and relabel the players so that $\sigma$ is the identity. We can explicitly find a subgame-perfect equilibrium by performing \emph{backward induction}. We first find the equilibria in the lowest subgames: we specify an action for the last player given any combination of actions of the other players. The resulting costs are known, so we can (and need to) specify one which minimises his cost. This defines $S_n$. (We have drawn in a unique arc from each decision node on layer $n$ to the leaves on layer $n+1$.) Given $A\in \prod_{i=1}^{n-2}\A_i$ and $S_n$, we can now pick a best response $S_{n-1}(A)$ for player $n-1$ to $S_n$ in the subgame induced by $A$. This defines $S_{n-1}$. (We draw in the next layer by minimising the cost of player $n-1$.) Continuing this way, we get a strategy profile which is a subgame-perfect equilibrium. \end{proof}
There might be multiple choices for best responses at each stage; for generic games  there is always a unique best response and in that case there is also a unique SPE. 

A subgame-perfect equilibrium $S$ is a strategy profile, hence there is an action profile $A(S)$ of $\Ga$ that corresponds to this: the \emph{outcome on the equilibrium path} of $S$.
\begin{definition}
Let $\Ga$ be a \Rosenthal game and $\sigma$ an order of the players. An action profile $A$ in $\Ga$ is called a 
\textbf{subgame-perfect outcome} (SPO) of $\Ga$ (corresponding to $\sigma$) if $A$ is the outcome on the equilibrium path of a subgame-perfect equilibrium of $\Ga^{\sigma}$.
\end{definition}
We denote the set of subgame-perfect outcomes of a game by $\SE$. We will usually leave the order on the players implicit. A statement such as ``in  symmetric singleton congestion games, each SPO is an NE'', should be read as ``for every  symmetric singleton congestion game $\Ga$, for every order on the players $\sigma$, every subgame-perfect outcome in $\Ga$ corresponding to $\sigma$ is stable''.

Paes Leme et al. \cite{leme} introduce the \textbf{Sequential Price of Anarchy} of a game $\Ga$ as 
\[
\SPoA(G)=\max_{A\in \SE(G)}\max_{B\in \prod_{i\in \N}\A_i} \frac{\sum_{i\in \N}c_i(A)}{\sum_{i\in \N}c_i(B)}
\]
where $\SE(G)$ is the set of subgame-perfect outcomes of $G$.

\newpage

\section{Missing proofs of Section~\ref{sec:sncg}}
\label{app3}

\paragraph{Full proof of Theorem \ref{thm:corNE}.}
Let $A$ be an NE of $\Ga$. We want to find a \emph{close} generic game $\Ga'$ of which $A$ is also an NE, and hence a subgame-perfect outcome of $\Ga'$. 

The game $\Ga'$ has the same set of players, resources and actions as $\Ga$, but its delay functions $d'_r$ are possibly different. We say $\Ga'$ is \textit{close} to $\Ga$ if for all paths $P\neq Q$ and congestion vectors $(x_r)$ and $(y_r)$
$$
c(P,x):= \sum_{r\in P} d_r(x_r) < c(Q,y) \implies c'(P,x):=\sum_{r\in P} {d'}_r(x_r) \leq c'(Q,y).
$$
We first show that we can find a close generic game $\Ga'$ which also has $A$ as equilibrium. Note that we can increase the cost of individual paths without effecting the other path costs, since each path contains a resource he shares with no other path, as shown by Milchtaic \cite{milnetw}. We need to ensure that for any $P\neq Q$ and congestion vectors $x,y$, $c'(P,x) \neq c'(Q,y)$. Starting off with the game $\Ga$ (a close game with $A$ as NE), we iteratively update the delay functions to make it generic ($c(P,x)\neq c(Q,y)$ for $P\neq Q$) while not violating the two conditions we already satisfy. The closedness condition can be ignored by always increasing or decreasing the cost $c(P,x)$ by an amount less than
$\delta := \min\{ |c(P,x)-c(Q,y)|: c(P,x)\neq c(Q,y), Q \text{ path}, y\text{ possible congestion vector}\}>0.$
(The set is finite, so we can take the minimum.) If $P \in A$, $x=x(A)$ and $y$ is the congestion vector when one person choosing $P$ in $A$ switches to $Q$, then the fact that $A$ is a Nash equilibrium means that $c(P,x)\leq c(Q,y)$. In order to keep the Nash equilibrium condition, if we encounter an equality of such form we simply choose to increase the cost of $Q$. (Since $y\neq x$, increasing $c(Q,y)$ never hurts the NE property of $A$.) Following these rules, we increase the delays of the resources unique to a path, until $c(P,x)\neq c(Q,y)$ for $P\neq Q$. This defines a close generic game $\Ga'$ that has $A$ as NE. By the previous results, $A$ is a subgame-perfect outcome of $\Ga'$ corresponding to some order $\sigma$.

Let $S$ be the corresponding SPE. This is also a strategy profile in the sequential-move version $\Ga^\sigma$. It remains to show it is subgame-perfect in $\Ga^\sigma$. By relabelling the players, we may assume $\sigma=\Id_{[n]}$. Let $A_1,\dots,A_\ell\in \A$ be given. We need to prove $Q=S_{\ell+1}(A_1,\dots,A_{\ell})$ is at least as good for player $\ell+1$ as any other $P \in \A$ (given that his successors follow $S$). Let $x$ be the congestion vector of the outcome when he plays $P$ and $y$ for when he plays $Q$. 
If $P\neq Q$, then $c'(Q,y)<c'(P,x)$ since $\Ga'$ is generic and $P$ is subgame-perfect according to $S$.
By the contraposition of the closedness property, we find $c(Q,y)\leq c(P,x)$. This shows $S$ is an SPE of $\Ga^\sigma$ as well, so that $A$ is an SPO of $\Ga$. \qed

\paragraph{Induced games.}
If we want to do induction on the graph structure for an $n$-player SNCG $\Ga$ on SP-graph $\Gr$, then we need to define induced games on the components of $\Gr$. 

If $\Gr$ consists of two SP-graphs $\Gamma_1$ and $\Gamma_2$ in series, then $\Ga$ \textit{induces} an SNCG $\Ga_i$ on $\Gr_i$ for $i=1,2$. Game $\Ga_i$ is played on graph $\Gamma_i$, which defines the set of resources and the common action space, and $\Ga_i$ inherits the delay functions and the set of players from $\Ga$. For an action profile $A$ of $\Ga$ and $i=1,2$, we denote
$$
\pi_i(A)=(A_j \cap  E(\Gamma_i) :j\in \N)
$$
where $A_j\cap E(\Gamma_i)$ denotes the pathsegment of $A_j$ that is contained in $\Gamma_i$.

If $\Gr$ consists of SP-graphs $\Gr_1$ and $\Gr_2$ in parallel, then $\Ga$ \textit{induces} games on $\Gr_1$ and $\Gr_2$ if we specify which players play in either game. The games inherit the delay functions from $\Ga$ again and the set of resources and common action set from the graph. If $A$ is an action profile in $\Ga$, then we let
$$
\pi_i(A) = (A_j: j\in \N, ~A_j\subseteq \Gr_i) ~~(i=1,2).
$$
The game $\Ga_i$ induced by $A$ on $\Gr_i$ is the game induced by $\Ga$ with $|\pi_i(A)|$ players, namely the players $\{j\in \N:A_j\subseteq \Gr_i\}$.

When graphs are composed in series, $O_A(\Ga)=O_{\pi_1(A)}(\Ga_1)+O_{\pi_2(A)}(\Ga_2)$. 
For graphs in parallel, 
$W_A(\Ga)=\max\{W_{\pi_1(A)}(\Ga_1), W_{\pi_2(A)}(\Ga_2)\}$ and $O_A(\Ga)=\min\{O_{\pi_1(A)}(\Ga_1),O_{\pi_2(A)}(\Ga_2)$.

\paragraph{Proof of Lemma \ref{lem:oppcosts}.}
We use the following result of \cite[Lemma 1]{optimalNEmakespan} in the proof.
\begin{lemma}
\label{lem:lastpaysworst}
Let $\Ga$ be a symmetric congestion game with $n$ players. Let $A$ be an action profile and $k$ a player so that $A_{-k}$ is a Nash equilibrium (in $\Ga^{n-1}$). Then $k$ pays $W_A(\Ga)$.
\end{lemma}
We prove for all series-parallel graphs $\Gamma$ that for any SNCG $\Ga$ on $\Gr$ with $m\geq n$ players, action profile $A$ of $\Ga^n$ and 1-lookahead outcome $B$ of $\Ga$, we have
$$
O_A(\Ga^n)\leq O_B(\Ga) \text{ and if }m>n \text{ then also } O_A(\Ga^n)\leq W_B(\Ga),
$$
by induction on the structure of $\Gr$. 
Note that by relabelling the players, we can assume without loss of generality that $B$ is a 1-lookahead outcome with respect to the identity.
If $\Gr$ is a single resource, then the statement holds. Suppose $\Gr$ consists of two SP-graphs $\Gamma_1$ and $\Gamma_2$ in series for which the statement has already been shown. Since $\pi_i(B)$ is also a 1-lookahead outcome in $\Gamma_i$ for $i=1,2$, we find by induction that
$$
O_A(\Ga^n)=O_{\pi_1(A)}(\Ga_1^n)+O_{\pi_2(A)}(\Ga_2^n)\leq O_{\pi_1(B)}(\Ga_1)+O_{\pi_2(B)}(\Ga_2)=O_B(\Ga)
$$
where $\Ga_i$ is the SNCG induced on $\Gamma_i$ for $i=1,2$. This proves the first statement. Suppose now $n<m$. Note that $B'=(B_1,\dots,B_{m-1})$ is a 1-lookahead outcome in $\Ga^{m-1}$ with $m-1\geq n$, so we can apply the first statement with $m'=m-1$ and Lemma \ref{lem:lastpaysworst} to find  $O_A(\Ga^n)\leq O_{B'}(\Ga^{m-1})= W_B(\Ga^m)$. 

Suppose $\Gr$ consists of two SP-graphs $\Gr_1$ and $\Gr_2$ in parallel (for which the statement has already been shown). For $i=1,2$, let $\Ga_i^{\ell}$ be the SNCG on $\Gr_i$ induced by $\Ga$ with $\ell$ players. Let $k= |\pi_1(A)|$ the number of players choosing from $\Gr_1$ in $A$. We have
$$
O_{A}(\Ga^n) = \min\{O_{\pi_1(A)}(\Ga_1^k),O_{\pi_2(A)}(\Ga_2^{n-k})\}.
$$
It hence suffices to upperbound either $O_{\pi_1(A)}(\Ga_1^k)$ or $O_{\pi_2(A)}(\Ga_2^{n-k})$.

Let $\ell$ be the number of players that choose paths in $\Gr_1$ in $B$ (our given 1-lookahead outcome). 
If $k=\ell$ and $n=m$ then we are done by induction. If this is not the case, then either $k<\ell$ or $n-k<m-\ell$, so we may assume $k<\ell$ by symmetry. Note that $\pi_1(B)$ is also a 1-lookahead outcome in $\Ga_1^\ell$, the SNCG induced on $\Gamma_1$, so that by induction we find
$$
O_{\pi_1(A)}(\Ga_1^{k})
\leq W_{\pi_1(B)}(\Ga_1^{\ell})\leq W_B(\Ga)
\leq O_B(\Ga),
$$
where the last equality is because $B$ is an NE. This proves both statements in this case. 
\qed

\paragraph{Proof of Proposition \ref{prop:optima}}
Assume there is a profile $B$ with $x(B)=x(A)$ for some $A\in \mathcal{M}$. Then $\Phi(A)=\Phi(B)$ so that $B\in \mathcal{M}$. By Corollary \ref{cor:brpsamepotential}, if this profile $B$ is a 1-lookahead outcome, all 1-lookahead outcomes are necessarily global minima of $\Phi$. It hence suffices to prove that for all $A\in \mathcal{M}$, there exists a 1-lookahead outcome $B$ with $x(A)=x(B)$.

We prove the statement by induction on the graph structure. For SNCG $\Ga^n$ with $n$ players on series-parallel graph $\Gamma$, let $\mathcal{M}(\Ga^n)$ denote the set of global minima of the game. When $\Gamma$ consists of a single resource, then the claim holds. If $\Gamma$ consists of two series-parallel graphs in series, then $\Ga^n$ induces games $\Ga_1^n$ and $\Ga^n_2$ on the components and 
$$
\mathcal{M}(\Ga^n)=\mathcal{M}(\Ga_1^n)\times \mathcal{M}(\Ga_2^n).
$$ 
Namely, if $\Phi_{\Ga^n}$ is the Rosenthal potential function corresponding to $\Ga^n$, then 
$$
\Phi_{\Ga^n}(A)=\Phi_{\Ga_1^n}(\pi_1(A))+\Phi_{\Ga_2^n}(\pi_2(A)),
$$
so that $A\in\mathcal{M}(\Ga^n) $ if and only if $\pi_i(A)\in \mathcal{M}(\Ga_i^n)$ for $i\in \{1,2\}$. If $A\in \mathcal{M}(\Ga^n)$, then for $i\in \{1,2\}$, $\pi_i(A)$ is a global minimum, hence by induction there is a 1-lookahead outcome $B_i$ inducing the same flows. Permuting the players in $B_1$ if necessary, we get a 1-lookahead outcome $(B_1,B_2)$ with the same congestion vector as $A$.

If $\Gamma_1$ and $\Gamma_2$ are placed in parallel, then $\Ga^n$ induces games $\Ga_1^k$ and $\Ga_2^{n-k}$ on the components for any $k\in \{0,\dots,n\}$. Let $A\in \mathcal{M}(\Ga^n)$ and $k=|\pi_1(A)|$. 
By induction, we find that $\pi_1(A)$ and $\pi_2(A)$ are 1-lookahead outcomes and it remains to prove that we can ``interweave" these into a 1-lookahead outcome $B$. Recall that each global minimum is a Nash equilibrium. 
This means that if $k=0,n$, then none of the players had a better response in the other component at any point, so that we are done with $B=A$. If this is not the case, we must have that $\pi_1(A)$ contains a cheapest path $P_1$ in $\Gamma_1$ (at the start) and $\pi_2(A)$ a cheapest path $P_2$ in $\Gamma_2$, so that either must be a cheapest path of $\Gamma$. Let $B_1$ denote this path. We now repeat this recursively. Remove the first occurrence of $B_1$ from the profile $A$ and apply the same reasoning for the game induced by $(B_1)$: again, either all players pick from the same component and we are done, or a cheapest path (with respect to the delay functions of the 1-player subgame induced by $(B_1)$) is chosen from both components. \qed


\newpage

\section{Generalization of $k$-lookahead outcomes to arbitrary games}

We elaborate on how our notion of $k$-lookahead outcomes can be generalized to arbitrary normal-form games. 
A \textit{(normal-form) game} $\Ga$ is a tuple $(\N,(\A_i)_{i\in \N}, (c_i)_{i\in \N})$ consisting of a finite set of players $\N=[n]$ and for each player $i\in \N$ a finite set of actions $\A_i$ and a cost function 
$c_i:\A_1\times\dots \times \A_n\to \R$.
Given a normal-form game $\Ga$ and an order $\sigma:\N\to [n]$ on the players, we need a cost function for player $i$
$$
c^{(j)}_{i}:\prod_{\ell \in \sigma^{-1}\{1,\dots,j\}} \A_\l \to \R
$$
that assigns each partial profile of actions of the players up to the $j$th player to a cost for player $i$ (for every $j\in \{k,\dots,n\}$ and $i \in \sigma^{-1}\{j-k+1,\dots,j\}$). 
For $j=n$, the game already prescribes cost functions, so it is natural to set $c^{(n)}_i=c_i$ for all $i\in \sigma^{-1}\{n-k+1,\dots,,n\}$. This is the minimal information required for defining a $k$-lookahead outcome of $\Ga$ corresponding to order $\sigma$. 

We want to refer to a $k$-lookahead outcome of $\Ga$ as a $k$-lookahead outcome corresponding to \emph{some} order of the players and we therefore define a $k$-lookahead structure so that it allows the definition of $k$-looakhead outcomes for any order on the players. 
\begin{definition}
\label{def:lookaheadstructure}
Let $\Ga=(\N,(\A_i)_{i\in \N},(c_{i})_{i\in \N})$ be a normal-form game. For $k\in \{1,\dots,n\}$, a \textbf{$k$-lookahead structure} for $\Ga$ is a tuple of cost functions
$$
\mathcal{C}_k=( c^{(j)}_i\mid  j\in \{k,\dots,n\}, i\in \N)
$$
where for $j\in \{k,\dots,n\}$ and $i \in \N$, the $j$th \textit{partial cost function} for player $i$
$$
c^{(j)}_{i}:\bigcup_{\substack{i\in N\subseteq \N \\ |N|=j}} \prod_{\ell \in N} \A_\l \to \R
$$
assigns each partial profile of size $j$ to a cost for player $i$. Moreover, $c^{(n)}_i=c_i$ for all $i\in \N$.
\end{definition}
For technical convenience, we define a $k$-lookahead structure for $k>n$ to be an $n$-lookahead structure. Note that a $k$-lookahead structure is also a $(k+1)$-lookahead structure. We will also refer to a $1$-lookahead structure as a \textit{lookahead structure}, since this allows the definition of $k$-lookahead outcomes for any $k$.

Each game has only one $n$-lookahead structure (given by the cost functions of the game). 
For $k\in \{1,\dots,n-1\}$, there exist many $k$-lookahead structures for every game $\Ga$.
\begin{example}
Let $\Ga$ be a normal-form game and $k<n$. For any $x\in \R$, we can define the ``constant $x$'' $k$-lookahead structure on $\Ga$: for $j\in \{k,\dots,n-1\}$, $i\in \N$ and $N$ a subset of players of size $j$ containing player $i$, let 
$c^{(j)}_{i}: (A_\ell)_{\ell \in N}\mapsto x.$
\end{example}
Ideally, we would like a $k$-lookahead structure to give a ``good prediction'' of the cost a player will end up getting.
We can use the set of costs 
$$
S=\left\{c_i(A_1,\dots,A_j,B_{j+1},\dots,B_n)\mid B_{j+1}\in \A_{j+1},\dots, B_n\in \A_n \right\}
$$ 
that player $i$ can possibly get in the future to give an intermediate cost for player $i$ (e.g. by setting $c_{i}^{(j)}(A):= \max(S)$). 
However, the size of the set $S$ can grow exponentially in the number of players, so that we lose the computational advantage of not expanding the game tree. Defining $k$-lookahead structures this way will therefore not be very useful in practice.

A $k$-lookahead structure also induces a $k$-lookahead structure on subgames, as shown below.
Let $\{c_i^{(j)}\}$ be a $k$-lookahead structure for a game $\Ga$. Let $A=(A_i)_{i\in [m]}$ give actions of players $1,\dots,m$. Then the subgame $\Ga' = (\N',(\A_i)_{i\in \N'}, (c'_i)_{i\in \N'})$ induced by $A$ has player set $\N'=\{m+1,\dots,n\}$ with cost functions
defined by $c'_i(B)=c_i(A,B)$ for $i\in \N'$. We can define a $k$-lookahead structure on $\Ga'$ by setting 
$$
{c'}_{i}^{(j)}:\bigcup_{\substack{i\in N\subseteq \N'\\|N|=j}}\prod_{\ell \in N} \A_\l \to \R:B\mapsto c_{i}^{(j)}(A, B)
$$
for $j\in \{k,\dots, n-m\}$ and $i\in \N'$.
For a game $\Ga=(\N,(\A_i)_i,(c_i)_i)$ with $k$-lookahead structure $\mathcal{C}_k=(c^{(j)}_i)_{i,j}$ for $k\leq n$ and any given order of the players $\sigma$, we can define a normal-form game 
$$
\Ga^k=\Ga^k(\sigma,\mathcal{C}_k):=(\N',(\A_i)_{i\in \N'}, (c^{(k)}_i)_{i\in \N'})
$$
by setting $\N'= \sigma^{-1}\{1,\dots,k\}$, the first $k$ players. If $k>n$, then we define $\Ga^k := \Ga^n$.
The definition of $k$-lookahead structure is chosen in such a way that $(\Ga')^k$ is defined for each subgame $\Ga'$. 
\begin{definition}
Let $\Ga$ be an $n$-player game and $\mathcal{C}_k$ a $k$-lookahead structure for $k\leq n$. 
An action profile $A$ is a \textbf{$k$-lookahead outcome} of $\Ga$ (with respect to $\mathcal{C}_k$) if there exists an order $\sigma$ on the players so that for each $i\in \N$ we have 
that $A_{i}$ equals the action $B_i$ played by player $i$ in some subgame-perfect outcome $B$  of $(\Ga')^{k}(\sigma, \mathcal{C}_k)$ 
that corresponds to the order $\sigma$, where $\Ga'$ is the subgame of $\Ga$  induced by $(A_j)_{\sigma(j)<\sigma(i)}$.
\end{definition} 
Hence, in a $k$-lookahead outcome player 1 chooses an action which is found by performing $k$ levels of backward induction while using the $k$-lookahead structure to assign costs to the nodes on level $k+1$. Player 2 performs $\min\{k,n-1\}$ levels of backward induction in the subtree with as root the node reached by following the arc corresponding to the action of player 1 (using the $k$-lookahead structure to assign costs to the nodes on level $\min\{k+2,n+1\}$).

The definition of $k$-$\LPoA$ and the relation between $n$-lookahead outcomes and subgame-perfect outcomes generalise when we call a game generic if $c_i(A)\neq c_i(B)$ for all profiles $A,B$ with $A_i\neq B_i$. 
The definition we gave for \Rosenthal games corresponds to the one above if the lookahead structure ``current costs'' is used: for $j\in \{1,\dots,n\}$ and $N\subseteq \N$ a set of size $j$, we set 
$$
c^{(j)}_i(A)=\sum_{r\in A_i}d_r(x(A)_r)~\text{ where }~x(A)_r=|\{\ell \in N:r\in A_\ell\}|,
$$
for $i\in N$ and $A\in \prod_{i\in N}\A_i$.

\end{document}